\documentclass[a4paper,10pt]{article}

\usepackage{cmap}           
\usepackage[utf8]{inputenc}    
\usepackage[T2A]{fontenc}       
\usepackage[english, russian]{babel}     
\usepackage{mathtext}           
\usepackage{amssymb}
\usepackage{amsthm}
\usepackage{amsmath}
\usepackage[pdftex,unicode]{hyperref}  
\usepackage{indentfirst}

\theoremstyle{plain}
\newtheorem{theorem}{Теорема}
\newtheorem{statement}{Утверждение}
\newtheorem{lemma}{Лемма}
\newtheorem{corollary}{Следствие}

\theoremstyle{definition}
\newtheorem{definition}{Определение}
\newtheorem{example}{Задача}

\DeclareMathOperator{\Vol}{Vol}

\newcommand{\M}{\mathbf{M}}
\newcommand{\D}{\mathbf{D}}

\hypersetup{colorlinks = true,
linkcolor = blue,
citecolor = green,
filecolor = magenta,
urlcolor = cyan,
pdftitle = Effective protocols for low-distance file synchronization,
pdfauthor = Aleksandr Chuklin}

\author{Чуклин А.Ю.\thanks{
    Московский Физико-Технический Институт}
}
\title{Эффективные протоколы
для задачи синхронизации слов с ограниченным хэмминговским расстоянием.}

\begin{document}

\maketitle

\begin{abstract}

Предположим, что на двух компьютерах хранятся файлы, которые в некотором смысле
похожи друг на друга. Требуется переслать один из этих файлов с первого компьютера
на второй, передавая по каналу связи как можно меньше данных.

В данной работе мы приводим обзор результатов, известных для данной коммуникационной задачи
для случая, когда рассматриваемые файлы ``близки'' друг к другу в смысле расстояния
Хэмминга. В основном мы систематизируем результаты известные ранее (полученные
разными авторами в 1990-х и 2000-х годах) и обсуждаем связь данной задачи с теорией
кодирования, алгоритмами хэширования и другими областями теоретической информатики.
В отдельных случаях предлагаются некоторые улучшения существующих конструкций.
\end{abstract}

\section{Введение}

Мы рассматриваем задачи об обмене информацией между двумя вычислительными
устройствами (для ясности скажем~--- между двумя компьютерами), соединёнными 
каналом связи. Пусть на этих  компьютерах хранятся файлы $X$ и $Y$ соответственно, 
и эти файлы не очень сильно отличаются друг от друга. Требуется переслать файл $X$ 
с  первого компьютера на второй, передавая по каналу связи как можно меньше данных. 
Разумеется, можно просто передать $X$ с первого компьютера на второй бит за битом.
Однако мы хотим воспользоваться тем, что на втором компьютере уже есть файл $Y$;
это даёт нам надежду, что возможен и более эффективный способ решения (требующий
передачи меньшего числа битов).

Данную задачу обычно называют \emph{задачей о синхронизации файлов}. Она 
представляет как теоретический, так и практический интерес. Чтобы предать 
задаче точный математический смысл, нужно сделать несколько уточнений. Следует указать,
в каком смысле файлы \emph{похожи} друг на друга. Кроме того, нужно уточнить, насколько
сложные вычисления могут производит оба компьютера, можно ли передавать по каналу связи 
данные  только в одном  направлении, или в обоих, из скольких раундов может состоять общение
между компьютерами и т.д. Различные способы уточнения условий задачи приводят к разным 
задачам \emph{теории коммуникационной сложности}.

Для большинства практических применений более адекватной кажется мера близости файлов, 
основанная на том или ином варианте редакторского расстояния. Задача синхронизации файлов с 
ограниченным редакторским расстояниям изучалась в~\cite{evfimievski,orlitsky-practical}, 
а также в~\cite{tridgell}.
Однако большинство известных результатов относятся к случаю,
когда  рассматриваемые файлы ``близки'' в смысле расстояния
Хэмминга (двоичные строки отличаются друг от друга в сравнительно небольшой доле позиций).
Хэмминговская метрика  кажется вполне естественной в контексте классической
теории кодирования. Неудивительно, что именно для хэмминговского отношения близости
для данной  коммуникационной задачи получены наиболее сильные и красивые результаты.

В данной работе мы ограничиваемся обзором результатов только для отношения
близости в смысле хэмминговского расстояния.
Мы систематизируем известные результаты для данной коммуникационной задачи
(для детерминированных и  вероятностных коммуникационных протоколов), обращая
особое внимание на вычислительную  сложность используемых алгоритмов.
В основном мы рассматриваем результаты,  известные по работам А.Орлитского
(\cite{orlitsky-worst-1}, \cite{orlitsky-worst-2}, \cite{orlitsky-average},
\cite{orlitsky}) и А.Смита (\cite{smith}).

В отдельных случаях мы приводим некоторые ранее не публиковавшиеся улучшения
конструкций коммуникационных протоколов, как например
протокол \ref{multiple NBA} для решения обобщённой задачи NBA или
альтернативное решение вероятностной задачи
синхронизации (утверждение~\ref{random_alt}).

\subsection{Исторический обзор и план статьи}
Понятие коммуникационной сложности было введено Э.Яо в 1979 году
в статье \cite{yao}.
В 1997 году вышла книга Э.Кушилевица и Н.Нисана \cite{kushilevitz-nisan}
с подробным изложением
постановок задач, методов и основных результатов данной области. Там же
был выделен такой раздел коммуникационной сложности, как
``общение с частичной информацией'' (\cite[Глава 4]{kushilevitz-nisan}),
к которому принадлежит изучаемая нами задача. Значительный вклад
в изучение задачи внёс А.Орлитский в своих статьях,
опубликованных в 1990-2000гг\footnote{
    В его работах используется термин ``интерактивное общение'' вместо термина,
    используемого Кушилевицем и Нисаном.
}. К наиболее поздним достижениям можно отнести результат А.Смита
\cite{smith} (2007 год) о вероятностном коммуникационном
протоколе. Позже идеи из этой работы были использованы в статье \cite{smith-guruswami}
2009 года для построения эффективных ``стохастических''
кодов исправляющих ошибки.

Данный обзор организован следующим образом. В начале мы обсуждаем постановку
задачи и детерминированный протокол её решения, оптимальный в
смысле числа передаваемых битов (разделы 2-3).
Затем мы переходим к построению вычислительно эффективного протокола (разделы 4-5).
В разделе 6 обсуждается вероятностная постановка задачи и различные алгоритмы
её решения. В разделе 7 приведена сводка полученных результатов и
сравнение различных алгоритмов.

Утверждения~\ref{lower-bound}, \ref{one-round-non-constructive} встречались
в разных работах и, по-видимому, не имеют единого автора.
Теорема~\ref{non-constructive sync}
о достижимости нижней оценки была впервые доказана А.Орлитским в 1993 году
в его работе~\cite{orlitsky}. Позже в 2003 году им же было доказана
теорема~\ref{one_not_the_best} о неоптимальности однораундового протокола.
Утверждение~\ref{constructive NBA} о конструктивном решении задачи
NBA неявно присутствовало в книге Э.Кушилевица и Н.Нисана
(\cite{kushilevitz-nisan}) 1997 года,
но не было ни доказано, ни чётко сформулировано. Утверждение~\ref{multiple NBA}
о модифицированной задаче NBA
рассматривается в данной работе впервые.
Утверждение~\ref{tricky} с использованием кодов,
исправляющих ошибки, было впервые доказано А.Орлитским в 1993 году
в той же работе~\cite{orlitsky}. Идея использовать коды, допускающие
декодирование списком, обозначена в работе~\cite{smith} 2007 года, но
никакого утверждения аналогичного~\ref{tricky-list} в работе не приводится.
Вероятностная постановка задачи вместе с доказательством теоремы~\ref{random}
исследуется в работе А.Смита~\cite{smith} 2007 года. В той же работе без доказательства
формулируется утверждение~\ref{strong converse}. Альтернативная
предложенной А.Смитом конструкция
вероятностного протокола (утверждение~\ref{random_alt})
и конструкция однораундового
вероятностного протокола с использованием кодов декодирования списком
(утверждение~\ref{random list-dec})
рассматриваются в данной работе впервые.

\section{Определения и постановка задачи}

Мы начнём с общего определения задачи коммуникационной сложности.
Участников коммуникационного протокола мы по традиции называем Алисой и Бобом.
Пусть у Алисы имеется строка $X \in \{0, 1\}^n$, а у Боба ---
строка $Y \in \{0, 1\}^m$. При этом Алисе и Бобу заранее известно, что
их входные данные (строки $X$ и $Y$) находятся в некотором отношении $S$.
Например, известно, что $n=m$ и хэмминговское расстояние между этими строками
не слишком велико.
Пусть также задана некоторая функция $f\colon \{0, 1\}^n \times \{0, 1\}^m \to \{0, 1\}^r$.
В теории коммуникационной сложности изучается следующая задача:\\

\noindent
\textit{Какое минимальное количество битов необходимо передать
от Алисы к Бобу и наоборот,
чтобы Боб узнал значение $f(X, Y)$?} \\

\noindent
При этом разрешается  передавать информацию в обоих направлениях,
но не требуется, чтобы Алиса узнала значение функции $f$.
Мы будем называть данный класс  задач \emph{коммуникационными задачами}.
(В монографии Кушилевица и Нисана \cite{kushilevitz-nisan} 
такой класс задач называется ``общением с частичной информацией'' ---
communication with partial information.)
Мы предполагаем, что Алиса и Боб заранее,
ещё не зная $X$ и $Y$, договариваются о протоколе общения. Затем, получив каждые
свою битовую строку, они обмениваются сообщениями 
в соответствии с оговорённым протоколом.
В результате этого диалога Боб должен определить значение $f(X,Y)$.
\emph{Коммуникационной сложностью} протокола общения
между Алисой и Бобом  мы называем количество передаваемых ими битов
(мы суммируем биты, переданные в обоих направлениях). Как правило, 
для получения верхних оценок на коммуникационную сложность задачи достаточно
подобного ``интуитивного'' определения (доказательство верхней оценки обычно состоит 
в предъявлении некоторого конкретного коммуникационного протокола). 
Однако для доказательства нижних
оценок (когда мы доказываем, что протоколов малой сложности не существует) требуется
иметь более формальное описание коммуникационного протокола:

\begin{definition}[Детерменированный коммуникационный протокол]\label{deterministic}
Де\-тер\-ме\-ни\-ро\-ван\-ным протоколом $\mathfrak{P}$
из $\{0,1\}^n \times \{0,1\}^m$ в $\{0,1\}^r$ называется
двоичное дерево, каждый внутренний узел $v$ которого помечен
либо функцией $a_v: \{0,1\}^n \to \{0, 1\}$, либо функцией $b_v: \{0,1\}^m \to \{0, 1\}$,
а каждый лист помечен функцией $z_l: \{0,1\}^m \to \{0,1\}^r$. Дополнительно
потребуем, чтобы вершины, предшествующие листьям,
были помечены функцией типа $a$, т.к. в нашей постановке окончательный ответ
должен получить Боб.

Значение протокола $\mathfrak{P}$ на входе $(X, Y)$ вычисляется обходом дерева
из корня до листа. В каждом внутреннем узле применяется соответствующая
функция ($a_v$ к $X$, $b_v$ к $Y$) и, в зависимости от результата, осуществляется
переход в левый узел (если значение функции 0) или в правый (если значение 1).
В листе Боб применяет функцию $z$ к своему $Y$ и получает значение $f'(X, Y)$,
которое должно совпадать с $f(X, Y)$ при любых $X$ и $Y$.

Сложностью данного протокола называется глубина дерева.
\end{definition}

Рассмотрим произвольный путь от корня к листу дерева. Выпишем порядок смены
типов функций, написанных в промежуточных узлах. Например, это 
$a, a, b, a, b, b, a, z$.
Посчитаем, количество смен $a$ на $b$ (или на $z$) и наоборот.
В нашем примере это число
равно 5, значит данное общение произошло в 5 раундов ---
в каждом раунде Алиса или Боб
могут послать сразу несколько битов.

\begin{definition}[Число раундов протокола]
Говорят, что протокол $\mathfrak{P}$ является $k$-раундовым,
если \textbf{любой}
путь от корня к дереву имеет не более $k$ смен типов функций.
\end{definition}

\begin{definition}[Коммуникационная сложность задачи]
Коммуникационной сложностью задачи называется минимальная сложность протокола
её решающего.
\end{definition}

В работах по теории коммуникационной сложности часто требуют, чтобы по окончании
коммуникационного протокола оба участника (Алиса и Боб) узнавали значение $f(X,Y)$.
Обращаем внимание читателя, что в нашем варианте определения только Боб обязан узнать
результат вычисления. Если $f(X,Y)$ состоит из большóго числа битов, то коммуникационные
сложности задачи могут существенно различаться в зависимости от того, требуется ли, чтобы
\textit{оба} участника узнавали ответ, или только один из них.

\bigskip

Введём теперь определение вероятностного протокола:

\begin{definition}[Вероятностный протокол]\label{probabilistic_private}
Как и в определении \ref{deterministic} протокол задаётся деревом. Только теперь
кроме строк $X$ и $Y$ имеются также
случайные строки $V_A \in \{0, 1\}^{r_A}$ (у Алисы)
и $V_B \in \{0, 1\}^{r_B}$ (у Боба),
и каждый внутренний узел помечен функцией
$a_v: \{0, 1\}^n \times \{0, 1\}^{r_A} \to \{0, 1\}$ либо
$b_v: \{0, 1\}^m \times \{0, 1\}^{r_B} \to \{0, 1\}$, а каждый лист ---
$z_l: \{0, 1\}^m \times \{0, 1\}^{r_B} \to \{0, 1\}^r$.
При этом мы требуем, чтобы протокол ошибался с вероятностью
не больше $\varepsilon$ (вероятность
берётся по случайным строкам $V_A$, $V_B$):
\begin{gather*}
    \forall (X, Y) \in S \;\;
    Pr_{V_A, V_B} \left[f'(V_A, V_B, X, Y) \ne f(X, Y) \right] \le \varepsilon
\end{gather*}
\end{definition}

\medskip

Приведём также другое определение, которое имеет меньший практический смысл, но бывает удобно
при доказательстве утверждений:

\begin{definition}[Протокол с общим источником случайности]
\label{probabilistic_common}
Как и в оп\-ре\-де\-ле\-нии \ref{probabilistic_private} 
мы имеем случайные строки у Алисы и Боба, только
на сей раз эта строка для них общая ($V_A = V_B = V$) и вероятность берётся по $V$:
\begin{gather*}
    \forall (X, Y) \in S \;\;
    Pr_{V} \left[f'(V, X, Y) \ne f(X, Y) \right] \le \varepsilon
\end{gather*}
В этом случае у Алисы и Боба имеется некоторая дополнительная общая информация, которая в некоторых
случаях приводит к существенному понижению коммуникационной сложности,
сохраняя прежнюю вероятность ошибки.
\end{definition}

Наша основная задача (о синхронизации файлов с ограниченным расстоянием Хэмминга)
легко формализуется с помощью определения \ref{deterministic}: 

\begin{example}[Задача о синхронизации файлов]\label{files}
Пусть $X$ и $Y$~-- битовые строки длины $n$, которые отличаются
не более чем в $\alpha n$ местах, т.е. $\rho(X, Y) \le \alpha n$, 
где $\rho(\cdot, \cdot)$ обозначает расстояние Хэмминга.
Формального говоря, $X$ и $Y$ связаны отношением
 $$S = \{(X, Y) \; | \; \rho(X, Y) \le \alpha n\}.$$
Требуется вычислить функцию $f(X, Y) = X$.
\end{example}

Прежде чем приступить к описанию простейших коммуникационных протоколов,
мы сформулируем ещё одну коммуникационную задачу (в дальнейшем
её решение нам потребуется в качестве промежуточного шага к решению нашей
основной  задачи~\ref{files}).

\begin{example}[Задача NBA]\label{nba}
Пусть $Y$ есть множество из $k$ строк\footnote{
    Мы позволяем себе некоторую вольность речи  и попеременно называем $Y$  строкой, набором строк 
    или чем-то ещё;  любой конструктивный объект легко закодировать в виде двоичной 
    строки, так что конкретный выбор представления входных данных не имеет большого значения.
}
длины $n$, а $X$ --- одна из строк, входящих в $Y$. 
Чтобы задать
коммуникационную задачу, остаётся определить функцию $f$. Мы положим 
$f(X, Y) = X$.
\end{example}

Общепринятое название задачи  NBA (от National Basketball Association ---
Национальная баскетбольная ассоциация Северной Америки) объясняется следующей метафорой.
Пусть Боб является баскетбольным болельщиком, который пропустил последнюю серию игр.
Вчера Алиса смотрела новости и знает какая команда победила, но не  знает
названия остальных команд, игравших в этой серии.
Задача состоит в том, чтобы передать название команды-победителя от Алисы к Бобу.

Наивное решение задачи NBA состоит в том, чтобы послать строку $X$ 
от Алисы к Бобу. Данное решение требует передачи $n$ битов информации. В дальнейшем
мы увидим, что при $k$ много меньшем $n$ задача имеет гораздо более эффективное решение.

\section{Детерминированные протоколы для задачи о синхронизации файлов}

В этой главе мы изучаем детерминированную коммуникационную
сложность для задачи~\ref{files}. Мы получим асимптотически
совпадающие верхние и нижние оценки и покажем, что  асимптотически
оптимальная коммуникационная сложность в этой задаче достигается
протоколами с 3 раундами. Дополнительно мы обсудим вопросы, связанные
с 1-раундовыми протоколами. При этом мы пока
никак не ограничиваем вычислительную сложность
алгоритмов Алисы и Боба (вычислительную сложность протоколов
мы рассмотрим подробно  в следующей главе).

\begin{statement}
\label{lower-bound}
Для каждого $\alpha <1/2$ детерминированная 
коммуникационная сложность  задачи~\ref{files} ограничена снизу величиной
\begin{gather*}
    H(\alpha) \cdot n + o(n),
\end{gather*}
где через $H(\cdot)$ обозначена двоичная энтропия:\\
$H(p) = - p \log p - (1 - p) \log (1 - p)$.
\end{statement}
\begin{proof}
Зафиксируем известное Бобу слово $Y$. При
этом в каждом листе $l$ дерева-протокола зафиксируется некоторое значение
$X_l$ --- возможный претендент на слово Алисы.
По условию задачи $X$ находится в  хэмминговском шаре
радиуса $\alpha n$ с центром в $Y$. Этот шар состоит из
 $\Vol(\alpha n, n) = \sum_{k=0}^{\alpha n} C_n^k$
точек.
Поскольку каждая строка длины $n$ из этого шара должна быть написана хотя бы
в одном из листьев, число листьев не может быть меньше $\Vol(\alpha n, n)$.
Высота дерева равна числу передаваемых
(в худшем случае) битов, следовательно,
сложность протокола не может быть меньше двоичного логарифма числа листьев
$\Vol(\alpha n, n)$. С помощью оценки Стирлинга нетрудно получить равенство
$\Vol(\alpha n, n) = 2^{H(\alpha)n + o(n)}$, что завершает доказательство утверждения.
\end{proof}
\begin{statement}\label{one-round-non-constructive}
Для любого $\alpha<1/2$ для задачи~\ref{files}
существует детерминированный однораундовый коммуникационный
протокол  сложности
$H(2 \alpha) \cdot n + o(n)$.
\end{statement}
\begin{proof}
Рассмотрим граф, вершинами которого будут все двоичные
строки длины $n$. Две вершины мы соединяем ребром,  если
хэмминговское расстояние между соответствующими строками не больше
$2 \alpha n$. Тогда степень каждой вершины будет равна
$\Vol(2 \alpha n, n) - 1$ (здесь, как и в предыдущем доказательстве,
$\Vol(r, n)$ обозначает число точек в хэмминговском шаре радиуса
$r$ в $\{0,1\}^n$).

Несложно показать, что любой граф степени $d$ допускает раскраску в $d+1$ цвет, так что все вершины,
соединённые ребром, покрашены в разный цвет
(можно окрашивать вершины графа в произвольном порядке;
для покраски очередной вершины может быть не более $d$ ограничений,
накладываемых покрашенными ранее соседями, так что всегда можно выбрать
цвет, не нарушающий условия).
Таким образом, интересующая нас раскраска существует. Мы фиксируем
одну такую раскраску и на её основе построим коммуникационный протокол.

По определению нашего графа,
любые вершины в пределах одного хэмминговского шара
радиуса $\alpha n$ оказываются
покрашены в разные цвета (т.к. расстояние между ними не больше диаметра шара,
т.е. $2 \alpha n$). Теперь несложно построить протокол.
Алиса посылает Бобу цвет своего слова $X$.
После этого Боб находит слово указанного цвета
в $\alpha n$-окрестности своего слова $Y$ и однозначно восстанавливает $X$.
При этом пересылается $\log (d + 1) = \log \Vol(2 \alpha n, n) 
= H(2 \alpha) \cdot n + o(n)$ битов, что и требовалось доказать.
\end{proof}

Возникает естественный вопрос. Можно ли существенно улучшить эту оценку?
В работе~\cite{orlitsky-one-way} А.Орлитским была доказана следующая теорема:

\begin{theorem}[А.Орлитский]\label{one_not_the_best}
Для любого $0 < \alpha < \frac{1}{4}$ существует $\beta > 0$, такое что для
достаточно большого $n$ сложность однораундового протокола будет не меньше,
чем $(1+\beta) \log \Vol(\alpha n, n)$.
\end{theorem}

Таким образом, достичь нижней грани на однораундовых протоколах невозможно.
С другой стороны, величина $\beta$, используемая в оценке, не очень велика:
при $\alpha$, стремящемся к нулю, $\beta$ стремится к нулю примерно с той же скоростью.

\bigskip

Далее мы покажем, что существует коммуникационный протокол, асимптотически достигающий
нижней оценки Утверждения~\ref{lower-bound}. Общение между Алисой и Бобом
в  этом протоколе будет происходить в три раунда.
Для построения этого протокола мы воспользуемся
локальной леммой Ловаса \cite{lovasz}:
\begin{lemma}[Локальная лемма Ловаса, симметричный случай]\label{lll}
\par Пусть \\
$A_1, A_2, \ldots, A_n$ --- события в произвольном вероятностном
пространстве. Пусть каждое событие $A_i$ взаимно независимо со всеми событиями,
за исключением не более чем $d$ событий $A_j$. Пусть также $P(A_i) \le p$ для
всех $1 \le i \le n$. Если при этом
\begin{gather*}
  e p (d + 1) \le 1,
\end{gather*}
то верно неравенство:
\begin{gather*}
  P\left(\bigwedge_{i=1}^n \overline{A_i}\right) > 0 \text{,}
\end{gather*}
т.е. с ненулевой вероятностью не выполнено ни одно из $A_i$.
\end{lemma}
Доказательство этой леммы можно найти, например, в~\cite{alon-spencer}.

\begin{theorem}[A.Orlitsky]\label{non-constructive sync}
Для задачи~1
существует трёхраундовый протокол с коммуникационной сложностью
$H(\alpha) \cdot n + o(n)$.
\end{theorem}
\begin{proof}
Протокол будет состоять из двух частей. В первой части Алиса посылает
Бобу некоторую хэш-функцию от своего слова 
(назовём её значение \textit{цветом} слова, 
чтобы не путать с другими хэш-функциями, которые появятся в доказательстве позже).
Боб определяет круг претендентов на слово $X$ (слов\'a такого же цвета, 
что и $X$, находящиеся в 
$\alpha n$-окрестности $Y$). Таких слов должно быть немного (в нашем примере
это будет некоторое число $k=O\left(\frac{n}{\log n}\right)$).
Во второй части Алиса с Бобом решают задачу NBA (наш пример~\ref{nba}
для указанного значения $k$).
\par Опишем подробно обе части протокола:
\begin{enumerate}
\item Докажем существование раскраски $\chi$ размера\footnote{
    $\Vol(\alpha n, n)$ --- объём хэмминговского шара радиуса $\alpha n$.
} $V = \Vol(\alpha n, n)$
такой, чтобы в каждом хэмминговском шаре $W_i$ радиуса $\alpha n$ слов каждого
цвета было не более $k$. Для этого случайным образом покрасим каждую
вершину в один из $|\chi|$ цветов.
\par Пусть $A_i$ обозначает противоположное событие:
в шаре $W_i$ найдётся цвет, представителей которого более, чем $k$.
Тогда верно неравенство:
\begin{gather*}
    P(A_i) \le C_V^k \cdot \frac{1}{|\chi|^k}\cdot |\chi|
\end{gather*}
Обозначим правую часть это неравенства через $p$.
\par Каждое из событий $A_i$ зависит не более, чем от $V_2 - 1$ других событий
$A_j$, где $V_2 = \Vol(2 \alpha n, n)$.
Запишем условие локальной леммы Ловаса (лемма \ref{lll}):
$$  e p V_2 \le 1  $$
Вспоминая, чему равен множитель $p$, перепишем это неравенство в виде
 $|\chi|^{k-1} \ge e C_V^k V_2$.
Напомним, что мы хотим найти раскраску из $V$ цветов. Таким образом,
условие леммы Ловаса принимает вид
$$
    V^{k-1} \ge e C_V^k V_2
$$
Далее мы преобразуем это неравенство к более удобному виду.
От замены $C_V^k$ на $\frac{V^k}{k!}$   неравенство
станет только сильнее. После сокращения получаем $ k! \ge e V V_2$.
Ещё раз огрубим наше неравенство, заменив $k!$ на $\left(\frac{k}{e}\right)^k$;
условие из леммы Ловаса принимает  вид
$\left(\frac{k}{e}\right)^k \ge e V V_2$. Воспользуемся тем фактом, что
$\log V = H(\alpha) n + o(n)$. После логарифмирования видим, что
достаточно потребовать выполнения неравенства
\begin{gather}\label{lovasz-fin}
    k (\log k - \log e) \ge (H(\alpha) + H(2 \alpha) + o(1)) n + \log e
\end{gather}
Положим $k = c \cdot \frac{n}{\log n}$. Теперь видно, что можно так выбрать
константу $c$, чтобы  неравенство (\ref{lovasz-fin}) выполнялось для
всех достаточно больших $n$.
\par Таким образом, для выбранного $k$ локальная лемма Ловаса утверждает, что существует
раскраска, при которой не выполнено ни одно из утверждений $A_i$ ---  что
нам и нужно. Поскольку существование  раскраски с нужными свойствами доказано,
можно найти одну из таких раскрасок перебором.
Эта раскраска фиксируется заранее
(она станет частью коммуникационного протокола).

Теперь  мы готовы описать первый раунда протокола. Он состоит в том, что 
Алиса посылает Бобу  цвет своей строки. Для этого требуется
передать $\log |\chi| = H(\alpha) n + o(n)$ битов.
\item \label{sync-NBA} После того, как Боба узнал цвет строки Алисы, в
его распоряжении имеется $k = c \cdot \frac{n}{\log n}$
претендентов на слово $X$. Мы оказываемся в условиях задача NBA. 
В \cite{kushilevitz-nisan} и \cite{vereschagin} можно найти
неконструктивные решения этой задачи.
Приведём наиболее простое из них.
\begin{lemma}\label{hash-family}
Существует семейство из $m = O(n k)$ хэш-функций
$$F \colon \{0,1\}^n \to \{1, \ldots, k^2\},$$ т.ч. для любого набора
из $k$ претендентов в этом семействе найдётся функция, которая не имеет коллизий
\textup(двух одинаковых значений\textup) на этом наборе.
\end{lemma}
\begin{proof}
Положим $b = k^2$. Выберем случайно $m$ функций из $\{0,1\}^n$ в $\{1, \ldots, b\}$
и покажем, что с положительной вероятностью такое семейство обладает
нужным нам свойством.

Для фиксированного набора претендентов $A$ вероятность того, что случайно
выбранная функция не имеет коллизий на этом наборе, не меньше
\begin{gather*}
    \frac{b}{b}\cdot\frac{b-1}{b}\cdot\ldots\cdot\frac{b-(k-1)}{b}
    \ge \left(1 - \frac{k}{b}\right)^k
    = \left(1 - \frac{1}{k}\right)^k \ge \frac{1}{4}
\end{gather*}
Это значит, что вероятность того, что все $m$ функций имеют коллизию
на данном наборе, не превосходит $\left(\frac{3}{4}\right)^{m}$. Остаётся
просуммировать данную вероятность по всем возможным наборам из
$k$ претендентов (таких наборов не более $2^{nk}$). Таким образом,
вероятность того, что случайно выбранный набор функций не удовлетворяет
условию леммы, не превосходит $2^{n k}\cdot \left(\frac{3}{4}\right)^{m}.$
При подходящем выборе $m$ (например, $m = 3 n k$) 
это число меньше единицы.
\end{proof}
Мы считаем, что некоторое семейство хэш-функций, удовлетворяющих условию
Леммы~\ref{hash-family}, 
найдено перебором (и будем использовать его в протоколе).

Второй раунд коммуникационного протокола состоит в том, 
что Боб находит подходящую 
хэш-функцию из семейства (без коллизий на множестве претендентов)
и шлёт её номер 
Алисе. В третьем раунде Алиса возвращает значение этой функции на своей строке. 
При этом от Боба к Алисе передаётся $\log m = O(\log n)$ битов,
а от Алисы к Бобу
$\log(k^2) = O(\log n)$ битов.
\end{enumerate}
Всего в протоколе передаётся $H(\alpha) n + o(n)$
битов, ч.т.д.
\end{proof}

Таким образом мы видим, что 3 раунда всегда достаточно для достижения
оптимума в смысле коммуникационной сложности. 

До сих пор мы интересовались
только коммуникационной сложностью протокола и не рассматривали вычислительную
сложность алгоритмов Алисы и Боба. Далее мы будем рассматривать
\emph{вычислительно эффективные} протоколы (детерминированные
и вероятностные).

\section{Эффективное решение задачи NBA}\label{hashing}

В этой части мы ненадолго отвлечёмся от задачи синхронизации
и обратимся к задаче NBA (пример~\ref{nba}).
Используя идеи из работы~\cite{fks},
мы построим \textbf{полиномиальный} протокол, который решает
эту задачу.
Полученный алгоритм
понадобится нам в дальнейшем для решения основной задачи.

Для начала сформулируем теоретико-числовую лемму:
\begin{lemma}[об асимптотике первой функции Чебышёва]\label{Chebyshyov}
Функция Чебышёва
\begin{gather*}
  \vartheta(x) = \sum_{p \le x, \text{ p-простое}} \log p
\end{gather*}
имеет следующую асимптотику:
\begin{gather*}
  \vartheta(x) = x + o(x), \text{ при } x \to \infty
\end{gather*}
\end{lemma}
\begin{proof}
Доказательство можно найти, например, в~\cite{dusart}.
\end{proof}

\begin{lemma}\label{U-size reduction}
Пусть дано множество $U$ размера\footnote{
    В наших задачах обычно размер множества $m=2^n$.}
$m$
и рассматриваются его всевозможные подмножества $S$ размера $k$.
Тогда существует семейство хэш-функций
$f_i\colon U \to \{1, \ldots, k^2 \log m\}$ размера не более $k^2 \log m$, 
такое что для
любого $S$ найдётся хэш-функция из этого семейства, которая не имеет коллизий 
на $S$ \textup(т.\,е. её ограничение на $S$ взаимнооднозначно\textup).
\par Если дополнительно потребовать, чтобы множество $U$ состояло из
$l$-значных чисел, то вычисление хэш-функции на элементе $x \in U$
производится за время $O(l^2)$.
\par При этом по списку элементов $S$ можно найти в данном семействе требуемую 
хэш-функцию за время $O(k^3 l^2 \log m)$.
\end{lemma}
\begin{proof}
Приведём доказательство, основанное на идеях из работы~\cite{fks}:
\par Будем считать элементы $U$ натуральными числами от $1$ до $m$.
Пусть $S = \{x_1, \ldots, x_k \}$. Обозначим
$t = \prod_{i > j} (x_i - x_j)  < m^{C_{k}^2}$.
Отсюда следует, что $\log t < C_{k}^2 \log m < k^2 \log m$.  Из 
Леммы~\ref{Chebyshyov} мы знаем, что
логарифм произведения простых чисел, меньших $N$, есть $N + o(N)$.
Отсюда следует, что обязательно найдётся простое число $q$,
меньшее $N = k^2 \log m$,
которое не входит в разложение $t$ на простые множители, т.е. не делит $t$.
Это означает, что функция $f_q (x) = x\bmod q$
является биекцией на $S$.\footnote{
    Через $x\bmod q$ будем здесь и далее обозначать остаток от деления
    $x$ на $q$.
}
При этом мы видим, что вычисление этой функции
не сложнее вычисления остатка по модулю, сложность которого для $l$-значных
чисел есть $O(l^2)$.
\par Для построения конкретного $q$ мы проделываем алгоритм решета
Эратосфена\footnote{ См.~\cite{Atkin2003}.} и находим все простые числа меньшие 
$N = k^2 \log m$ (за время $O(N / \log \log N)$). Затем для каждого
простого числа проверяем, является ли полученная функция взаимно-однозначной
(на это требуется $O(k l^2)$ времени и $O(k l)$ памяти).
Итак, суммарное время ограничено величиной 
$O(N k l^2) = O(k^3 l^2 \log m)$, ч.т.д.
\end{proof}

\begin{statement}\label{constructive NBA}
Пусть Алиса и Боб решают задачу \ref{nba} (NBA), в которой у Боба
есть $k$ строк длины $n$ каждая. Для этой задачи существует полиномиальный
двухраундовый протокол
с коммуникационной сложностью 
\begin{gather*}
C \le 2 \log (k^2 n) +  O(1) = 2 \log n + 4 \log k + O(1)
\end{gather*}
При этом вычислительная сложность для Боба есть $O(k^3 n^3)$,
а для Алисы $O(n^2)$.
\end{statement}
\begin{proof} 
\par \textbf{РАУНД 1.} Боб подбирает простое число $q$,
существование которого гарантирует лемма \ref{U-size reduction},
и посылает это $q$  Алисе.
\par \textbf{РАУНД 2.} Алиса вычисляет значение $f_q (X) = X\bmod q$
и посылает его Бобу. Боб вычисляет значение $f_q(\cdot)$ от всех своих
$X_i$ и определяет алисин $X$.
\par В этом протоколе в каждую сторону передаётся число битов равное
$\log q \le \log (k^2 n) + O(1)$.
\par Вычислительная сложность для Боба состоит в поиске 
требуемого простого числа $q$
(сложность $O(k^3 n^2 \log 2^n) = O(k^3 n^3)$)
и в вычислении хэш-функции от всех
своих $k$ строк (сложность\footnote{
    Вспомним, что
    вычисление остатка для $n$-значных чисел происходит за время $O(n^2)$
}
$O(k n^2)$).
Для Алисы же это лишь вычисление одной
хэш-функции от своей строки --- $O(n^2)$.
\end{proof}
Заметим, что при простой передаче слова от Алисы к Бобу нам потребовалось
бы передать экспоненциально бóльшее количество битов.
\begin{corollary}\label{non-constructive refinement}
Часть \ref{sync-NBA}  (второй и третий раунды) в протоколе из утверждения \ref{non-constructive sync}
может быть осуществлена конструктивно.
\par Коммуникационная сложность
этой части будет $O(\log n)$, а вычислительная сложность --- 
$O\left(\frac{n^6}{\log^3 n} + \frac{n^3}{\log n}\right) = o(n^6)$.
\end{corollary}
\begin{proof}
Для этого нужно решать задачу NBA как показано в
утверждении~\ref{constructive NBA}.
\end{proof}

\medskip

В качестве бонуса рассмотрим обобщение задачи \ref{nba} (NBA),
при котором у Алисы есть не одна строка длины $n$, а $l$ строк (напомним,
что у Боба в наличии имеется $k$ строк) и Боб хочет узнать все эти $l$ строк.
Например, Алиса знает названия команд, прошедших в четвертьфинал и хочет
сообщить их Бобу.
\par Очевидным решением было бы такое,
при котором Алиса и Боб просто $l$ раз подряд применяют протокол из
утверждения~\ref{constructive NBA}. Коммуникационная сложность такого
протокола была бы $ C \le (2 \log n + 4 \log k) \cdot l + O(l)$.
Проблема этого решения в том, что от Алисы к Бобу передаются слишком большие
хэш-значения.
Покажем, что для $n$, много больших $k$, мы можем улучшить эту оценку.
Для этого, правда, нам придётся перейти к вероятностному протоколу
(определение~\ref{probabilistic_private}) с нулевой вероятностью
ошибки $\varepsilon = 0$.

\begin{lemma}[О вторичном хэшировании] \label{secondary-hash}
Пусть дано множество $U'$ размера $m' = v-1$ (v --- простое)
и рассматриваются его 
всевозможные подмножества $S'$ размера $k$. Рассмотрим семейство отображений
$h_s[x] = (s x\bmod v)\bmod 2 k^2$, где $0 \le s \le m'$ задаёт номер функции
внутри семейства. Тогда если мы зафиксируем
некоторое конкретное множество $S'$, то как минимум половина функций из этого
семейства не будет иметь коллизий на $S'$.
\end{lemma}
\begin{proof}
Доказательство данного утверждения можно найти в \cite[Corollary 4]{fks}.
\end{proof}

Теперь мы можем сформулировать утверждение 
(сравните с утверждением~\ref{constructive NBA}):
\begin{statement}\label{multiple NBA}
Для рассмотренной модификации задачи NBA существует протокол с коммуникационной
сложностью
$$
C \le 2 \log n + (2 l + 4) \log k + O(l)
$$

Этот коммуникационный протокол \underline{вероятностный}: в алгоритме
Боба используются случайные биты. При этом ответ, получаемый в
результате выполнения протокола, всегда корректен (вероятность ошибки нулевая),
а число пересылаемых битов не зависит от выпавших случайных битов.

Среднее время работы для Боба равно $O(k^3 n^3)$
(усреднение берётся по случайным битам).
Алиса использует детерменированный алгоритм, который выполняется за
время $O(l n^2)$.
\end{statement}
\begin{proof}
\par \textbf{РАУНД 1.} Боб подбирает простое число $q$,
существование которого гарантирует лемма \ref{U-size reduction}.
После этого мы от исходного множества $U$ размера $2^n$ переходим
к множеству вычетов по модулю $q$, т.е. к множеству $U'$ размера
$k^2 \cdot n$. Далее Боб начинает случайно выбирать номер $s$ функции
из семейства $\{h_s\}$, о котором идёт речь в
лемме~\ref{secondary-hash}, и проверяет, что данная функция не даёт коллизий.
Среднее время, которое он потратит
на эту процедуру, прежде чем достигнет успеха, будет равно:
\begin{gather*}
\sum_{i=1}^{\infty}\frac{1}{2^i} \cdot i \cdot
O(k \cdot n^2) = 2 \cdot O(k \cdot n^2) = O(k \cdot n^2)
\end{gather*}
\par После этого Боб посылает Алисе пару чисел $(q, s)$, что составляет
$2 \log (k^2 \cdot n) + O(1)$ битов.
\par \textbf{РАУНД 2.} Алиса, получив от Боба числа $q$ и $s$, вычисляет
от своих $l$ строк функцию $x \to h_s[x\bmod q]$, и посылает эти $l$ значений
Бобу. Размер каждого значения --- $\log(2 k^2)$
(в силу леммы~\ref{secondary-hash}) и вычисляется оно за время $O(n^2)$.
\par Легко видеть, что этот алгоритм удовлетворяет требуемым оценкам.
\end{proof}
Итак, используя идеи хэширования, мы получили решение
задачи~\ref{nba} вместе с её обобщением, а также продвинулись на пути
к получению конструктивного протокола для задачи~\ref{files}.

\section{Связь с теорией кодирования}
\subsection{Коды, исправляющие ошибки}
В утверждении~\ref{one-round-non-constructive} был предъявлен неконструктивный
однораундовый протокол с коммуникационной сложностью
$C = H(2\alpha) \cdot n + o(n)$. Можно ли сделать его конструктивным?
Оказывается, что в некоторых случаях это можно сделать, используя коды,
исправляющие ошибки
\footnote{В этой
главе мы будем пользоваться стандартной терминологией теории
кодирования. Определения кода, линейного кода, кодового расстояния, и
т.п., см., например, в~\cite{sudan}.}.
\begin{definition}
Кодом с параметрами $(\alpha, R(\alpha))$ будем называть код, исправляющий
часть ошибок, равную $\alpha$, и имеющий скорость $R(\alpha)$
(т.е. слово длины $n$ кодируется словом длины $\frac{n}{R(\alpha)}$).
\end{definition}

Рассмотрим две конструкции, основанные на кодах, исправляющих ошибки:

\begin{statement}[<<Грубая конструкция>>]\label{brute}
Пусть имеется семейство систематических кодов\footnote
{ Систематическим называется код, для
которого кодируемое слово всегда является префиксом своего кода.
В частности, любой линейный код можно переделать в систематический.
Обратное, вообще говоря, не верно.},
с параметрами $(\beta, R(\beta))$.
Тогда по нему можно построить \textbf{однораундовый} коммуникационный протокол
решения задачи~\ref{files} с коммуникационной
сложностью $C = \left(\frac{1}{r} - 1\right) \cdot n$,
где $r$ --- корень уравнения\footnote{Несложно проверить, что для
$R(x) = 1 - H(2 x)$ корень указанного уравнения существует и единственен,
при условии что $\alpha \in \left[0, \frac{1}{2}\right]$.}
\begin{gather*}
  r = R(\alpha r)
\end{gather*}
При этом вычислительная сложность полученного алгоритма равна
вычислительной сложности алгоритма кодирования/декодирования
соответствующего кода.
\end{statement}
\begin{proof}
Пусть имеются строки $X$ у Алисы и $Y$ у Боба как в задаче~\ref{files}.
Рассмотрим следующий протокол:\\
\textbf{Алиса.} Алиса применяет к своему $X$ алгоритм кодирования для кода с
$\alpha r$ ошибками. При этом получается кодовая строка длины
$\frac{n}{R(\alpha r)} = \frac{n}{r}$, первые $n$ битов которой совпадают
с битами $X$. Следующим своим шагом Алиса посылает Бобу оставшиеся
$\frac{n}{r} - n$ битов.\\
\textbf{Боб.} Боб получает от Алисы $\frac{n}{r} - n$ битов и приписывает
их в конце своего $Y$. После этого он трактует полученное слово, как кодовое
слово Алисы с ''ошибками``, которые как будто бы внёс канал с шумом.
Доля таких ''ошибок`` в слове будет равна
$\frac{\alpha \cdot n}{\frac{n}{r}} = \alpha r$, и с ними успешно справится
алгоритм декодирования для используемого кода. Результатом работы алгоритма
будет являться алисино слово $X$ с приписанными в конце битами.

Корректность протокола очевидна.
\end{proof}
\par \textbf{Замечание.} Код, используемый в этой конструкции не обязан
быть линейным --- подойдёт любой систематический код.

\par Рассмотрим также другую конструкцию, предложенную А.Орлитским
в работе~\cite{orlitsky}:
\begin{statement}[<<Тонкая конструкция>>, Орлитский]\label{tricky}
Пусть имеется код с параметрами $(\alpha, R(\alpha))$, обладающий
следующими свойствами:
\begin{enumerate}
\item Код является линейным.
\item \label{tricky-2}
$\forall y \in \{0, 1\}^n$ существует эффективный способ найти кодовое
слово $z$, т.ч. $\rho(y, z) \le \alpha n$ (т.е. существует эффективный
алгоритм декодирования).
\end{enumerate}

Тогда по данному коду строится \textbf{однораундовый} коммуникационный протокол с коммуникационной
сложностью $C = \left(1 - R(\alpha)\right) \cdot n$.
Вычислительная сложность данного протокола
есть $\text{СЛОЖНОСТЬ КОДА} + O(n^3)$, где через ''СЛОЖНОСТЬ КОДА``
обозначена вычислительная сложность алгоритма декодирования соответствующего
кода, исправляющего ошибки.
\end{statement}
\begin{proof}
Пусть $H$ -- матрица проверок на чётность для данного кода\footnote{
Т.е. кодовыми словами являются те и только те $z$, для которых
$H z = 0.$}.
Матрица $H$ имеет размер $n \times (n - k)$, где $k = R(\alpha) n$.

\textbf{Алиса.} Алиса вычисляет синдром $h = H X$ для своего слова $X$
и посылает его Бобу.\\
\textbf{Боб.} Получив синдром $h$ от Алисы, Боб решает систему линейных
уравнений $H t = h - H Y$
и находит \textit{какое-то} решение $t = y'$ данной системы.
На это у него уйдёт время $O(n^3)$. 
Далее Боб применяет протокол для нахождения строки $z$, ''декодирующей``
$y'$ (свойство \ref{tricky-2}). Тогда $H z = 0$ и $\rho(z, y') \le \alpha n$.
Пусть $z^* = y' + Y - X$.
Тогда $H z^* = H y' + H Y - H X = h - HY + HY - h = 0$,
$\rho(z^*, y') = \rho(y' + Y - X, y') = \rho(Y - X, 0) = \rho(Y, X) \le \alpha$.
Выходит, что $z^*$ --- кодовое слово, обладающее теми же свойствами, что и
$z$. В силу однозначности декодирования получаем, что $z^* = z$, а значит
Боб может вычислить $X$ следующим образом: $X = y' + Y - z$.
\par При этом от Алисы к Бобу передаётся $n - k$ битов, то есть
коммуникационная сложность данного протокола равна
$C = n - k = \left(1 - R(\alpha)\right) \cdot n$, ч.т.д.
\end{proof}

\subsection{Коды, допускающие декодирование списком}

В утверждении~\ref{non-constructive sync} был приведён неконструктивный
протокол для решения задачи~\ref{files} синхронизации файлов.
В разделе~\ref{hashing}
мы изучали как сделать конструктивным второй этап протокола синхронизации
(этап, на котором решается задача NBA). Возникает естественный вопрос: можно
ли сделать конструктивным также и первый этап протокола, не увеличив число
передаваемых битов? Частично ответить на этот вопрос
помогают коды, допускающие декодирование списком. Для этого заметим,
что конструкция из утверждения~\ref{tricky} допускает обобщение на случай
кодов, допускающих
декодирование списком (list-decoding codes)\footnote{См.~\cite{sudan},
\cite{Guruswami2006}}:

\begin{statement}\label{tricky-list}
Пусть имеется семейство кодов, допускающих декодирование списком,
с параметрами $(\alpha, R(\alpha), L)$, т.е. таких кодов, что в
$\alpha n$ - окрестности \textit{любого} слова из $\{0, 1\}^n$
находится не более $L$ кодовых слов.
Потребуем также выполнения следующих условий:
\begin{enumerate}
\item Код является линейным кодом, допускающим декодирование списком.
\item $L = O(poly(n))$, т.е. $L$ не превосходит некоторого полинома от $n$.
\item \label{tricky-list-3}
$\forall y \in \{0, 1\}^n$ существует эффективный способ найти список
не более чем из $L$ слов $z_i$,
т.ч. $\rho(y, z_i) \le \alpha n$ (т.е. существует эффективный
алгоритм декодирования списком).
\end{enumerate}
\par Тогда по данному коду строится \textbf{трёхраундовый}
коммуникационный протокол с коммуникационной
сложностью $C = (1 - R(\alpha)) n + o(n)$.
Вычислительная сложность данного протокола
есть $\text{СЛОЖНОСТЬ КОДА} + O(L^3 \cdot n^3)$.
\end{statement}
\begin{proof}
Пусть $H$ --- вновь матрица проверок на чётность для данного кода.
Протокол устроен следующим образом:\\
\textbf{РАУНД 1.} Алиса вычисляет синдром $h = H X$ для своего слова $X$
и посылает его Бобу.\\
\textbf{РАУНДЫ 2-3.} Получив синдром $h$ от Алисы, Боб решает систему линейных
уравнений $H t = h - H Y$
и находит \textit{какое-то} решение $t = y'$ данной системы.
На это у него уйдёт время $O(n^3)$. 
Далее Боб применяет протокол для нахождения строк $z_i$, ''декодирующих``
$y'$ (свойство \ref{tricky-list-3}).
Тогда $H z_i = 0$ и $\rho(z_i, y') \le \alpha n$.
Пусть $z^* = y' + Y - X$.
Тогда $H z^* = H y' + H Y - H X = h - HY + HY - h = 0$,
$\rho(z^*, y') = \rho(y' + Y - X, y') = \rho(Y - X, 0) = \rho(Y, X) \le \alpha$.
Выходит, что $z^*$ --- кодовое слово, обладающее теми же свойствами, что и
все $z_i$. Значит $z^*$ обязательно содержится в списке слов $z_i$,
следовательно $X$ содержится в списке слов $y' + Y - z_i$. Длина этого
списка $L = O(poly(n))$, а значит можем применить протокол из
утверждения~\ref{constructive NBA} для решения
задачи~\ref{nba} (задачи NBA).
\par В общей сложности в этом протоколе от Алисы к Бобу передаётся
$n - k + O(\log n)$ битов,
т.е. $C = (1 - R(\alpha)) n + o(n)$. Вычислительная сложность складывается
из сложности решения системы уравнений, сложности кодирования/декодирования
для кода, исправляющего ошибки и сложности решения задачи NBA. Суммарная
вычислительная сложность в силу следствия~\ref{non-constructive refinement}
равна:
$\text{СЛОЖНОСТЬ КОДА} + O(L^3 \cdot n^3)$.
\end{proof}

Таким образом, если мы научимся строить эффективный код,
допускающий декодирование списком,
со скоростью, близкой к границе Хэмминга ($R(\alpha) = 1 - H(\alpha)$),
то мы автоматически получим оптимальный эффективный трёхраундовый протокол
решения задачи~\ref{files}. Тут следует отметить, что теоретически такие
коды существуют (см., например, \cite{Guruswami2002}, \cite{ZyablovPinsker}),
и даже выбранный случайно линейный код с большой вероятностью будет
удовлетворять нужным нам требованиям (см.\cite[теорема 5.3]{Guruswami2001}).
Но алгоритм декодирования списком для такого кода будет экспоненциальным от $n$,
т.к. для этого требуется перебрать все кодовые слова\footnote{
    Заметим, что сложность декодирования будет именно
    $O\left(2^n \cdot \textup{poly}(n)\right)$,
    а не больше. Это будет важно нам при рассмотрении
    теоремы~\ref{random}.
}. В настоящее время неизвестно конструкций с полиномиальным алгоритмом
декодирования.

\section{Вероятностная постановка задачи}

\par В этом разделе мы изучим вероятностные протоколы решения
задачи о синхронизации файлов:
\begin{example}[Задача о вероятностной синхронизации файлов]
\label{files-random}
Пусть у Алисы имеется строка из $n$ битов $X$, а у Боба --- строка битов $Y$ такой же длины. 
Заранее известно, что $\rho(X, Y) \le \alpha n$. Мы рассматриваем коммуникационные
протоколы, в котором Алисе и Бобу разрешается использовать источники случайных битов 
(см.определение~\ref{probabilistic_private}).
По выполнении протокола Боб должен восстановить слово $X$ с вероятностью ошибки
не более $\varepsilon$.
\end{example}

\subsection{Нижняя оценка}
Начнём с несложного утверждения, которое даёт нижнюю оценку на 
коммуникационную сложность:
\begin{statement}
\label{strong converse}
Для любого положительного $\varepsilon < \frac{1}{2}$
коммуникационная сложность данной задачи не меньше, чем
$H(\alpha) \cdot n + o(n)$.
\end{statement}
\begin{proof}
Вместо коммуникационного протокола с отдельными источниками случайности
мы докажем нижнюю оценку для сложности протокола
с общим для Алисы и Боба источником случайных битов
(определение~\ref{probabilistic_common}).
Любой протокол с раздельными источниками случайности легко переделать в
протокол с такой же коммуникационой сложностью и общим источником
случайности. Роль случайной строки $r$ в этом протоколе будет
играть конкатенация строк $r_A$ и $r_B$ из протокола с раздельными случайными
битами. Поэтому минимальная сложность  для протоколов с общедоступными случайными битами
может оказаться только меньше, чем в исходной модели.

Зафиксируем произвольным образом слово Боба (некоторую строку $Y$ из $n$ битов).
Вероятностный протокол с общим источником случайности можно рассматривать
как распределение вероятностей на некотором
семействе $D$ детерминированных коммуникационных протоколов.
Напомним, что как только мы зафиксировали слово $Y$, детерминированный
протокол представляет собой двоичное дерево,
каждый лист которого помечен
одним из возможных ответов (т.е., какой-то строкой $X \in \{0, 1\}^n$),
а сложность протокола есть глубина этого дерева. Отметим, что все деревья
из семейства $D$ имеют одинаковую глубину, а различаются лишь пометки в вершинах.

По условию задачи, Алиса может получить в качестве своего входа любое слово 
$X\in\{0,1\}^n$, находящееся на расстоянии не более $\alpha n$ от выбранного $Y$.
Напомним, что число таких слов равно $N=2^{H(\alpha) \cdot n + o(n)}.$
Для каждого такого $X$  протокол выдаёт правильный ответ  с вероятностью 
не менее $1 - \varepsilon>1/2$. Это значит, что вероятностная мера протоколов в $D$,
имеющих хотя бы один лист с данной пометкой $X$, должна быть не меньше $1/2$.
Проделав это рассуждение для каждого из $N$ возможных $X$, получим, что
у протокола из семейства $D$ число листьев не может быть меньше $N/2$.

Остаётся заметить, что сложность детерминированного протокола не меньше, чем
логарифм числа его листьев. Следовательно, сложность протокола из $D$
не  меньше $\log (N/2) = H(\alpha) \cdot n + o(n)$, а значит
и сложность исходного вероятностного протокола не меньше
$ H(\alpha) \cdot n + o(n)$,
ч.т.д.
\end{proof}


\subsection{От 3-раудового детерминированного протокола к 1-раундовому вероятностному протоколу}

Переделаем построенный нами детерменированный трёхраундовый протокол,
основанный на кодах, допускающих декодирование списком,
в однораундовый протокол для вероятностной задачи.
Для начала докажем аналог леммы \ref{U-size reduction}:

\begin{lemma}\label{primes-random}
Пусть дано множество $U = \{1, \ldots, 2^n\}$ и некоторое его
подмножество $S$ размера k. Тогда для любого
$a \in \mathbb{N}$ существует
число $A = a n k^2 \log (a n k^2) + o(a n k^2)$,
т.ч. равномерно выбранное \textbf{простое} число $q$ с отрезка $[1, A]$
с вероятностью $1 - \frac{1}{a}$ задаёт
функцию $f_q: x \to x\bmod q$, не имеющую коллизий на $S$, .
\par Как и ранее, вычисление хэш-функции происходит за $O(n^2)$.
Выбор функции из этого семейства (выбор простого числа)
требует времени 
$O\left(\frac{A}{\log \log A}\right)$ (решето Эратосфена).
\end{lemma}
\begin{proof}
Как и в лемме \ref{U-size reduction} обозначим 
$S = \{x_1, \ldots, x_k \}$,
$t = \prod_{i > j} (x_i - x_j)  < m^{C_{k}^2}$.
Отсюда следует, что $\log t < C_{k}^2 \log m < k^2 \log m = n k^2$.
Значит число $t$ имеет не более $n k^2$ простых делителей.
Но на отрезке $[1, A]$ имеется не менее $a n k^2$ простых чисел
в силу теоремы о распределении простых чисел\footnote{
    См., например, работу~\cite{Selberg1949}.
}. Значит, выбрав
наугад простое число из этого интервала, мы с вероятностью
$1 - \frac{1}{a}$ попадём в число, которое не является делителем
$t$, а значит не даёт коллизий на $S$, ч.т.д.
\end{proof}

Теперь сформулируем утверждение о существовании протокола:
\begin{statement}\label{random list-dec}
Пусть имеется семейство кодов, допускающих декодирование списком, с параметрами
$(\alpha, R(\alpha), L)$, удовлетворяющее всем условиям утверждения
\ref{tricky-list}. Тогда для любого $a$, растущего медленнее чем $2^n$
($a = \Omega(1)$, $a = 2^{o(n)}$),
по данному коду строится
\textbf{однораундовый} вероятностный протокол с коммуникационной
сложностью $C = (1 - R(\alpha)) n + o(n)$ и вероятностью ошибки
не превосходящей $1 - \frac{1}{a}$.
Вычислительная сложность данного протокола, равна:
\begin{gather*}
    \text{СЛОЖНОСТЬ КОДА} + O(a n L^2 \log(a n L^2)) + O(L n^2) + O(n^3)
\end{gather*}
\end{statement}
\begin{proof}
Рассмотрим следующий протокол:
\par \textbf{Алиса}. Алиса, как и ранее, вычисляет синдром $h = H X$
своего слова $X$. Кроме этого, она генерирует простые числа
с отрезка $[1, A]$ и выбирает одно из них наугад
(см. лемму \ref{primes-random}). После этого, она хэширует $X$
при помощи выбранного ей простого числа $q$ и посылает
хэш-значение Бобу вместе с синдромом $h$ и простым числом $q$.
\par \textbf{Боб}. Боб получает всё от Алисы и сначала запускает
алгоритм восстановления списка претендентов на слово $X$
аналогично утверждению \ref{tricky-list}. После этого он
смотрит на полученное от Алисы число $q$. Боб проверяет, правильно
ли оно хэширует его список претендентов. Если нет, он сигнализирует
об ошибке синхронизации. В противном случае, он смотрит
какой из его претендентов даёт хэш-значение, совпадающее со значением,
полученным от Алисы, и выводит его в качестве ответа.
\par Коммуникационная сложность состоит из длин синдрома
($(1 - R(\alpha)) n$ битов), простого числа $q$ и хэш-значения $f_q[X]$
($2 \log A = O(\log (a n k^2)) = O(\log a) + o(n) = o(n)$).
\par Вычислительная сложность данного протокола складывается из
сложности кода, сложности решения системы уравнений ($O(n^3)$),
сложности построения списка простых чисел
(решето Эратосфена, 
$O(A / \log \log A) = O(A) = O(a n L^2 \log(a n L^2))$) и
сложности вычисления хэш-значений для всех претендентов на слово
$X$ ($O(L n^2)$).
\end{proof}

\subsection{Протокол А.Смита}

У протокола из предыдущего раздела есть очевидный недостаток: для его
работы по прежнему нужно иметь хороший код, допускающий декодирование списком.
С другой стороны, в статье~\cite{smith} А.Смитом была явно построена конструкция
протокола решения задачи~\ref{files-random}
о вероятностной синхронизации файлов
(с полиномиальными алгоритмами для Алисы и Боба), в котором
достигается  нижняя оценка из утверждения~\ref{strong converse}:

\begin{theorem}[A.Smith, Theorem 4]\label{random}
Для всякого $\delta = \Omega\left(\frac{\log \log n}{\sqrt{\log n}}\right)$
существует однораундовый
протокол решения задачи~\ref{files-random}
с коммуникационной сложностью $H(\alpha)\cdot n + O(\delta n)$
и ошибкой $\varepsilon \le 2^{-\delta^3 n / (24 \cdot \log n)}$.
При этом время работы алгоритма ограничено многочленом от $n$.
\end{theorem}

Доказательство данной теоремы можно найти в оригинальной работе~\cite{smith} А.Смита.
Здесь же мы подробно рассмотрим конструкцию, которая была предложена А.Ромащенко
и несколько отличается от оригинальной конструкции А.Смита.
Эта конструкция кажется нам более простой для изложения
и выглядит более практичной с точки зрения
возможный применений, т.к. все вычисления ограничены полиномом небольшой
степени.
Следует, однако, отметить, что оценка на вероятность ошибки
протокола, полученного в утверждении~\ref{random_alt}, будет существенно хуже,
чем у А.Смита:
$\varepsilon \le O\left(\frac{1}{\sqrt{\log n}}\right)$.

Как обычно, мы считаем,
что Алиса и Боб в качестве входных данных получают $n$-битные строки
$X$ и $Y$ соответственно. В дальнейшем в описании конструкции мы
будем использовать несколько параметров, зависящих от $n$:
целое число $k$ (мы выберем его порядка $\log n$),
целое число $s$ (оно будет $o(n/k)$)
и рациональное число $\delta$ (оно будет стремиться к нулю при росте $n$).
Более точные значения этих параметров мы уточним по ходу доказательства.

Сначала мы опишем сам коммуникационный протокол. Он
будет состоять из трёх этапов:

\medskip

\textbf{Первый этап.}
Для начала введём следующее определение:\footnote{
    Заметим, что А.Смит в своей конструкции использовал более
    сложное определение ``почти независимости''.
    Для этого ему пришлось прибегнуть к технике
    так называемых KNR генераторов (\cite{knr}). Подробнее об этом в
    его работе~\cite{smith}.
}
\begin{definition}[Семейство попарно независимых перестановок]
\label{pseudo_random_family}
Семейство $\mathfrak{S}$
перестановок $\pi \ :\ \{1,\ldots,n\} \to \{1,\ldots,n\}$ 
будем называть семейством попарно независимых
перестановок, если перестановки $\pi \in \mathfrak{S}$ обладают
следующими свойствами:
\begin{itemize}
\item \textit{Равномерность}:
    $\forall i, x$ $Pr_{\pi \in \mathfrak{S}}[\pi(i) = x]
        = \frac{1}{n}$.
\item \textit{``Почти независимость''}:
    $\forall i \ne j$
    случайные величины $\pi(i)$ и $\pi(j)$ ``почти независимы'',
    т.е. $\forall x \ne y \; Pr_{\pi \in \mathfrak{S}}[\pi(i) = x, \pi(j) = y]
       = \frac{1}{n (n-1)}$.
\end{itemize}
\end{definition}

Несложно проверить, что если в качестве $\mathfrak{S}$ 
рассмотреть все $n!$ перестановок на $n$ элементах, то оба условия будут
выполнены. О том как построить такое семейство \textit{небольшого размера}
мы поговорим чуть позже.

Итак, в начале общения Алиса и Боб договариваются о некоторой перестановке
$\pi$ и применяют эту перестановку к битам своих слов $X$ и $Y$: в
полученных словах $X'$ и $Y'$ для каждой позиции $i$ значение бита
$X'(i)$ (соответственно, $Y'(i)$) вычисляется как $X(\pi(i))$
(соответственно, $Y(\pi(i))$).

\medskip

\textbf{Второй этап.} Алиса и Боб разбивают свои слова
(подвергнутые перестановке битов)
на блоки длины $k$. Таким образом, слова $X'$ и $Y'$ представляются
в виде конкатенаций
$$
 X' = X'_1\cdots X'_m,\  Y' = Y'_1\cdots Y'_m,
$$
где каждый из блоков $X'_i$, $Y'_i$ состоит из $k$ битов, и $m=n/k$.

Далее для каждой пары $X'_i, Y'_i$  для $i=1,\ldots,m$ Алиса и Боб выполняют
протокол синхронизации корректный для слов, отличающихся не более чем
в $(\alpha+\delta)k$ позициях (мы закладываем здесь небольшой запас $\delta$,
выбор которого мы уточним ниже).

Такой выбор внутреннего протокола даёт основания утверждать, что Бобу удастся
восстановить правильно те
слова $X'_i$, которые отличаются от соответствующего $Y'_i$ менее, чем в 
$(\alpha+\delta) k $ позициях.
Про другие пары блоков ($X'_i$ и $Y'_i$, которые отличаются
не менее, чем в $(\alpha+\delta) k $ позициях),
внутренний протокол может выдавать ошибочный ответ.
В дальнейшем  пары блоков $X'_i$ и $Y'_i$, которые отличаются
не менее, чем в $(\alpha+\delta)k $ позициях, мы будем называть \emph{опасными}.

Обсудим подробнее устройство внутреннего протокола.
Проще всего было бы использовать неэффективный метод из
теоремы~\ref{non-constructive sync}. Но тогда для выбора раскраски
потребовалось бы перебрать $2^{\Omega(k 2^k)}$ вариантов, и
для обеспечения полиномиального времени работы алгоритма мы были бы вынуждены
выбрать $k = O(\log \log n)$.
Но, как мы увидим позже, такое маленькое $k$ не подходит для нашей конструкции.

Мы можем позволить себе б\'oльшее $k$,
если в качестве внутреннего кода будем использовать метод из
утверждения~\ref{tricky-list}. В качестве базового кода для этого утверждения
мы возьмём код, который впервые неявно встречался
в работе~\cite{ZyablovPinsker} 1982 года.
Позже, в работе В.Гурусвами (2001 год) было явно сформулировано и доказано
следующее утверждение~\cite[Theorem 5.3]{Guruswami2001}:
\begin{statement}[Зяблов, Пинскер]\label{zyablov}
Для любого $p \le \frac{1}{2}$ и любого $L \to \infty$ существует
семейство линейных кодов, допускающих декодирование списком,
с длиной списка $L$ и скоростью
$R \ge 1 - H(p) - o(1)$.
Если матрица для кода выбирается случайно и равномерно, то вероятность
того, что она задаёт код с нужными свойствами, будет не меньше, чем
$1 - 2^{-\Omega(\sqrt n)}$.
\end{statement}

Таким образом, Алиса выбирает матрицу с нужными параметрами
случайным образом и посылает её Бобу.
При нашем выборе $k = \log n$ сложность декодирования списком\footnote{
    Именно здесь нам важно, что декодирование происходит за
    $O\left(2^k \cdot \textup{poly}(k)\right) = O(n^2)$.
} для
данного кода будет $O(n^2)$, а число передаваемых битов ничтожно мало.

Обозначим за $Y''_1,\ldots,Y''_m$ те блоки из $k$ битов, которые Боб будет иметь в 
результате выполнения этого этапа.

\textbf{Третий этап.}
Алиса интерпретирует свои блоки $ X'_1\cdots X'_m$ как набор из $m$ 
элементов конечного поля $\mathbb{F}$ (размер поля равен $2^k$; обозначим
его элементы $\mathbb{F}=\{a_1,a_2,\ldots,a_{2^k}\}$). Алиса находит многочлен
$P(t)$  степени не выше $m-1$ над этим полем такой, что 
 $$
 P(a_1) = X'_1,\ \ldots,\ P(a_m) = X'_m
 $$
Затем она  вычисляет значения найденного многочлена $P$ в точках поля
$a_{m+1},\ldots, a_{m+s}$ и посылает их Бобу. В этом месте важно, что
размер поля $2^k > m + s = n/k + s$, --- именно поэтому мы не могли взять
$k = O(\log \log n)$.

По существу, Алиса посылает
синдром для кода Рида-Соломона от строки $X'_1\cdots X'_m$. \footnote{
    За подробным определением кода
    Рида-Соломона и его свойств мы отсылаем читателя к~\cite{sudan}
    или любой другой книге по теории кодирования.
}

Боб сравнивает полученные значения многочлена $P$ с блоками $Y''_1, \ldots, Y''_m$
и восстанавливает значения $X''_1,\ldots,X''_m$. Таким образом,
можно сказать, что Боб выполняет
процедуру декодирования для кода Рида--Соломона. Эта процедура
происходит корректно, если число ошибок $N < s/2$.

Далее Боб составляет конкатенацию
полученных  блоков
 $$
 X'' = X''_1 \cdots X''_m
 $$
и применяет к битам этого слова перестановку, обратную к $\pi$.
Ниже мы покажем, что (при разумном выборе параметров)
полученный результат с большой вероятностью совпадает с $X$.

\bigskip


\textbf{Оценка коммуникационной сложности.}
Сложность первого этапа зависит от размера семейства
перестановок $\mathfrak{S}$. Как мы увидим позже, размер этого
семейства будет $n(n-1)$, т.е. на этом этапе достаточно
передать $O(\log n)$ битов.

На втором этапе для каждого блока $i=1,\ldots,m$ Алиса и Боб передают
$H(\alpha+\delta)k + o(k)$ битов. Таком образом, общая коммуникационная
сложность этого этапа равна $H(\alpha+\delta)n + o(n)$. Если выбирать значение
параметра $\delta$ так, что $\delta\to 0$, то мы получаем коммуникационную
сложность $H(\alpha)n + o(n)$.

На третьем этапе Алиса пересылает $s$ элементов поля $\mathbb{F}$, что требует
в общей сложности $sk$ битов. Если выбирать значение параметра $s$ так, что
$sk/n \to 0$, то коммуникационная сложность данного этапа есть $o(n)$.

Таким образом, коммуникационная
сложность всего протокола равна $H(\alpha)n + o(n)$.

\medskip


\textbf{Сложность вычислений.}
Сложность выбора случайной перестановки $\pi \in \mathfrak{S}$,
как мы увидим позже, заключается лишь в доступе к генератору
случайных чисел и передаче кода перестановки ($O(\log n)$).

На втором этапе происходит протокол из утверждения~\ref{tricky-list},
при этом, как мы видели, сложность декодирования списком ограничена
величиной $O(n^2)$. Значит суммарная сложность этого этапа
есть $\left(O(n^2) + \textup{poly}(\log(n))\right) \cdot m = O(n^3)$.

На третьем этапе Алиса производит обычную полиномиальную интерполяцию
(что можно сделать за полиномиальное время). Боб же восстанавливает
многочлен по его значениям, некоторые из которых заданы с ошибками.
Это можно сделать за полиномиальное время ($O(n^2)$)
с помощью алгоритма Берлекампа-Мэсси, \cite{Atti2006}.
Отметим, что вычисления Алисы и Боба на третьем этапе являются
соответственно кодированием и декодированием для кода Рида-Соломона.

Таким образом, вычисления Алисы и Боба требуют полиномиального времени.

\medskip

\textbf{Вероятность ошибки.}
Как мы уже говорили (утверждение~\ref{zyablov}),
вероятность выбора неправильного кода, допускающего
декодирование списком, экспоненциально мала. Мы не будем учитывать её
в приведённых ниже рассуждениях, а вместо этого домножим окончательную
вероятность корректной синхронизации на $\left(1 - 2^{-\Omega(\sqrt n)}\right)$.

На втором этапе протокола ошибки при передаче блоков $X'_i$ от Алисы к Бобу
могут возникнуть только в опасных парах блоков.
Обозначим количество опасных пар $N$. Заметим, что $N$ есть случайная величина,
зависящая от выбора (на первом этапе) случайной перестановки $\pi$.

Далее, на третьем этапе Боб должен восстановить  многочлен степени
не более $m-1$. Бобу даны значения этого многочлена в $m+s$ точках
($m$ значений $Y''_i$, полученных на втором этапе протокола и ещё $s$ значений, 
присланных Алисой на третьем этапе). Среди первых
$m$ значений максимум $N$ могут быть ``неправильными'' (если пара ``опасная'', то блок $Y''_i$
может не совпадать с $X'_i$).
Таким образом,  Бобу нужно восстановить многочлен степени
не более $m-1$,
если даны его значения в $m+s$ точках, среди которых не более $N$
значений могут оказаться неправильными.
Если $N < s/2$, то искомый многочлен определен однозначно.

\begin{lemma}\label{lemma-unif-pi}
Пусть исходные слова $X$ и $Y$ отличаются не более чем в $\alpha n$
позициях, и слова $X'$ и $Y'$ получены из них с помощью
попарно независимой
перестановки позиций $\pi$.
Тогда вероятность того, что число опасных блоков окажется больше $s/2$,
не превосходит
$\frac{n}{s k^2 \delta^2}$.
\end{lemma}

\begin{proof}
Пусть $\xi_i$ --- случайная величина, индикатор события \\
$\{X(\pi(i)) \ne Y(\pi(i))\}$.
Обозначим через $a_1, \ldots, a_r$ различающиеся позиции в исходных
строках $X$ и $Y$ ($r \le \alpha n)$.
Используя определение~\ref{pseudo_random_family},
найдём величины $\M\xi_i$ и $\M\xi_i \xi_j$ для $i \ne j$:
\begin{gather*}
    \M\xi_i = \sum_{l=1}^r Pr[\pi(i) = a_l] = \frac{r}{n} \\
    \M\xi_i \xi_j = \sum_{l=1}^r \sum_{t=1}^r Pr[\pi(i) = a_l, \pi(j) = a_t]
        = \sum_{t \ne l} \frac{1}{n(n-1)} = \frac{r(r-1)}{n(n-1)}
\end{gather*}

Тогда вероятность того, что $l$-й блок ``опасный'' равна:
\begin{gather*}
    Pr\left[\sum_{j=0}^{k-1} \xi_{l k + j} \ge (\alpha + \delta) k \right]
\end{gather*}

Обозначим $\Phi_l = \sum_{j=0}^{k-1} \xi_{l k + j}$ и найдём его
математическое ожидание и дисперсию:
\begin{gather*}
    \M\Phi_l = k \M\xi_i = \frac{kr}{n} \\
    \D\Phi_l = \M\Phi_l^2 - (\M\Phi_l)^2
        = k \M\xi_i^2 + k(k-1) \M\xi_i\xi_j - \frac{k^2 r^2}{n^2} = \\
    = \frac{k r}{n} + k(k-1)\frac{r(r-1)}{n(n-1)} - \frac{k^2 r^2}{n^2} = \\
    = \frac{k r}{n} \left(
        1 + \underline{\frac{k r}{n-1}} - \frac{k}{n-1} - \frac{r}{n-1}
        + \frac{1}{n-1} - \underline{\frac{k r}{n}}
        \right) = \\
    = \frac{k r}{n} \left(1 + \frac{k r}{n(n-1)} +
        \left(\frac{r}{n(n-1)}- \frac{r}{n}\right)
        + O\left(\frac{k}{n}\right)\right)
        = [\text{т.к. } r \le \alpha n] = \\
    = k \cdot \frac{r}{n} \left(1 - \frac{r}{n} +
        O\left(\frac{k}{n}\right)\right)
        \le \frac{k}{4} + O\left(\frac{k^2}{n}\right)
        \text{, т.к. } r \le \alpha n \le n / 2
\end{gather*}

Воспользуемся неравенством Чебышёва:
\begin{gather*}
    Pr[l\text{-й блок ``опасный''}] 
        = Pr[\Phi_l \ge (\alpha + \delta) k ] = \\
    = Pr\left[\Phi_l - \M\Phi_l
            \ge \left(\alpha - \frac{r}{n} + \delta\right) k
        \right] \le \\
    \le Pr[\Phi_l - \M\Phi_l \ge \delta k]
    \le Pr[|\Phi_l - \M\Phi_l| \ge \delta k]
        \le \frac{\D \Phi_l}{(\delta k)^2} \le \\
    \le \frac{1}{4\delta^2 k} + O\left(\frac{1}{\delta^2 n}\right)
        \le \frac{1}{4\delta^2 k}
            \left(1 + O\left(\frac{k}{n}\right)\right) \le \\
    \le \frac{1}{2\delta^2 k} \text{ при достаточно большом } n
\end{gather*}

Обозначим через $\varphi_l$ индикатор того, что $l$-й блок опасный.
Оценим нужную нам вероятность по неравенству Маркова:
\begin{gather*}
    Pr\left[\sum_{l=1}^{m} \varphi_l \ge s/2\right]
        \le \frac{\sum_{l=1}^m M\varphi_l}{s/2}
        \le \frac{2}{s} \cdot \frac{n}{k} \cdot \frac{1}{2 \delta^2 k}
        = \frac{n}{s k^2 \delta^2} \text{, ч.т.д.}
\end{gather*}

\end{proof}

\bigskip

Теперь мы можем сформулировать утверждение аналогичное теореме~\ref{random}:

\begin{statement}\label{random_alt}
При
$k = \log n$, $s = \frac{n}{\log n \cdot \log{\log n}}$,
$\delta = \sqrt{\frac{\log {\log n}}{\sqrt{\log n}}}$
ошибка в протоколе синхронизации стремится к нулю с ростом $n$,
коммуникационная сложность равна $H(\alpha) n + o(n)$,
а все вычисления Алисы и Боба требуют времени $O(n^3)$.
\end{statement}

Для того, чтобы доказать это утверждение нам осталось построить
семейство попарно независимых перестановок небольшого размера.

Для начала заметим, что, не меняя асимптотик, можем считать
число $n$ простым. Для этого воспользуемся следующим фактом из
теории чисел (\cite{Baker2001}):

\begin{lemma}[Бэйкер, Харман, Пинтц]
Для любого $n$, большего некоторого $n_0$, интервал $[n - n^{0.525}, n]$ содержит простое
число.
\end{lemma}

Используя эту теорему, мы можем обе строки $X$ и $Y$ дополнить
нулями до ближайшего простого числа $p \ge n$.

Теперь, в качестве семейства $\mathfrak{S}$ мы можем рассмотреть функции
$f_{a,b}: i \to (a i + b) \mod n$,
где $a \in \{1, \ldots, n-1\}$, $b \in \{0, \ldots, n-1\}$. Легко проверить,
что это семейство удовлетворяет определению~\ref{pseudo_random_family}.

Таким образом, перестановка $\pi$ задаётся двумя числами $a, b$,
которые может сгенерировать
Алиса в самом начале протокола и послать Бобу.
При этом, она перешлёт всего $2 \log n$ битов.

Итак, мы получили конструкцию протокола с раздельным источником случайности,
оптимального с точки зрения коммуникационной сложности
и работающего за полиномиальное время. Этот протокол похож на протокол, предложенный А.Смитом,
но в отличие от него не требует построения KNR-генератора\footnote{
    Подробнее см. в~\cite{smith},\cite{knr}
}(что обеспечивает лучшую вычислительную сложность),
зато имеет худшую вероятность ошибки.

\section{Таблицы результатов}\label{summary}

В заключение приведём сводку того, что получено на данный момент
для задачи~\ref{files}
о синхронизации файлов:

\medskip
{
\footnotesize
\begin{tabular}{|p{2cm}|p{2.2cm}|p{3cm}|p{3cm}|}
\hline
\multicolumn{4}{|c|}{\textbf{Детерминированные однораундовые протоколы}}\\
\hline

& ком\-му\-ни\-ка\-ци\-он\-ная сложность 
& вычислительная сложность
& замечания \\
\hline
оценка снизу
& $H(\alpha) n + o(n)$
& \hrule
& см.утв.~\ref{lower-bound} \\
\hline
оценка сверху
& $H(2 \alpha) n + o(n)$
& \textbf{экс\-по\-нен\-ци\-аль\-ная}
& см.утв.~\ref{one-round-non-constructive}.
Получение лучшей верхней оценки --- открытый вопрос.\\
\hline
<<грубая конструкция>>
& меньше, чем $H(2 \alpha) n + o(n)$
& \textbf{по\-ли\-но\-ми\-аль\-ная} (равна СЛОЖНОСТИ КОДА)
& см.утв.~\ref{brute} \\
\hline
<<тонкая конструкция>> (А.Ор\-лит\-ский)
& $H(2 \alpha) n + o(n)$
& \textbf{по\-ли\-но\-ми\-аль\-ная} (равна СЛОЖНОСТИ КОДА + $O(n^3)$)
& см.утв.~\ref{tricky} \\
\hline
\end{tabular}
\bigskip

\begin{tabular}{|p{2cm}|p{2cm}|p{3cm}|p{3cm}|}
\hline
\multicolumn{4}{|c|}{\textbf{Детерминированные многораундовые протоколы}}\\
\hline

& ком\-му\-ни\-ка\-ци\-он\-ная сложность 
& вычислительная сложность
& замечания \\
\hline
оценка снизу
& $H(\alpha) n + o(n)$
& \hrule
& см.утв.~\ref{lower-bound} \\
\hline
оценка сверху (А.Ор\-лит\-ский)
& $H(\alpha) n + o(n)$
& \textbf{экс\-по\-нен\-ци\-аль\-ная}
& см.утв.~\ref{non-constructive sync} \\
\hline
<<тонкая конструкция>>
& $H(\alpha) n + o(n)$
& \textbf{по\-ли\-но\-ми\-аль\-ная} (равна СЛОЖНОСТИ КОДА + 
$O(L^3 \cdot n^3)$)
& см.утв.~\ref{tricky-list}. Для этого нужен полиномиальный
код декодирования списком, близкий к границе Хэмминга.
Построение такого кода пока является открытым вопросом.\\
\hline
\end{tabular}

\bigskip

\begin{tabular}{|p{2cm}|p{2cm}|p{3cm}|p{3cm}|}
\hline
\multicolumn{4}{|c|}{\textbf{Вероятностные протоколы}}\\
\hline

& ком\-му\-ни\-ка\-ци\-он\-ная сложность 
& вычислительная сложность
& замечания \\
\hline
оценка снизу
& $H(\alpha) n + o(n)$
& \hrule
& см.утв.~\ref{strong converse} \\
\hline
кон\-струк\-ция А.Смита
& $H(\alpha) n + o(n)$
& \textbf{по\-ли\-но\-ми\-аль\-ная}
($O(n^3)$ для утв.~\ref{random_alt})
& см. теорему~\ref{random} и утв.~\ref{random_alt}\\
\hline
де\-ко\-ди\-ро\-ва\-ние списком
& $H(\alpha) n + o(n)$
& \textbf{по\-ли\-но\-ми\-аль\-ная}
(равна СЛОЖНОСТИ КОДА + $O(a n L^2 \log(a n L^2)) + O(L n^2) + O(n^3)$) 
& см.утв.~\ref{random list-dec}. Для этого нужен полиномиальный
код декодирования списком, близкий к границе Хэмминга.
Построение такого кода пока является открытым вопросом.\\
\hline
\end{tabular}
}  



\bibliographystyle{utf8gost71u}
\bibliography{collection}

\end{document}